\documentclass[english,runningheads,envcountsame,a4paper]{llncs}
\usepackage[T1]{fontenc}
\usepackage[latin9]{inputenc}
\usepackage{array}
\usepackage{float}
\usepackage{amsmath}
\usepackage{amssymb}
\usepackage{graphicx}

\makeatletter

\newcommand{\noun}[1]{\textsc{#1}}
\providecommand{\tabularnewline}{\\}

\usepackage{amsmath}
\usepackage{rotating}
\usepackage{ifpdf} 
\usepackage{epstopdf}
\usepackage{hyperref}
\hypersetup{hidelinks}

\usepackage{caption}

\usepackage{complexity}

\DeclareMathOperator{\rng}{rng}

\DeclareMathOperator{\OM}{\Omega}

\providecommand{\openbox}{\leavevmode
  \hbox to.77778em{%
  \hfil\vrule
  \vbox to.675em{\hrule width.6em\vfil\hrule}%
  \vrule\hfil}}
\makeatletter
\DeclareRobustCommand{\dispqed}{%
  \ifmmode
    \eqno \def\@badmath{$$}
    \let\eqno\relax \let\leqno\relax \let\veqno\relax
    \hbox{\openbox}%
  \else
    \leavevmode\unskip\penalty9999 \hbox{}\nobreak\hfill
    \quad\hbox{\openbox}%
  \fi
}

\makeatother

\usepackage{babel}
\begin{document}
\title{Parameterized Complexity of Synchronization and Road Coloring} 
\titlerunning{Parameterized Complexity of Synchronization and Road Coloring}

\author{Vojt\v{e}ch Vorel\inst{1}\thanks{Corresponding author. Supported by the Czech Science Foundation grant GA14-10799S.} \and Adam Roman\inst{2}\thanks{Supported by the Polish Ministry of Science and Higher Education Iuventus Plus grant IP2012 052272.}} 

\authorrunning{Vojt\v{e}ch Vorel, Adam Roman} 
\tocauthor{Vojt\v{e}ch Vorel, Adam Roman} 

\institute{Faculty of Mathematics and Physics, Charles University, Malostransk\'{e} n\'{a}m. 25, Prague, Czech Republic,\\ \email{vorel@ktiml.mff.cuni.cz}, 
\and Institute of Computer Science, Jagiellonian University,  Lojasiewicza 6, 30-348 Krakow, Poland, \\ \email{roman@ii.uj.edu.pl}}
 
\maketitle  
\begin{abstract} 

First, we close the multivariate analysis of a canonical problem concerning short reset words (SYN), as it was started by Fernau et al. (2013). Namely, we prove that the problem, parameterized by the number of states, does not admit a polynomial kernel unless the polynomial hierarchy collapses.
Second, we consider a related canonical problem concerning synchronizing road colorings (SRCP). Here we give a similar complete multivariate analysis. Namely, we show that the problem, parameterized by the number of states, admits a polynomial kernel and we close the previous research of restrictions to particular values of both the alphabet size and the maximum word length.

\end{abstract}

\section{Introduction}

Questions about synchronization of finite automata has been studied
since the early times of automata theory. The basic concept is very
natural: For a given machine, we want to find an input sequence that
would get the machine to some particular state, no matter in which
state the machine was before. Such sequence is called a \emph{reset
word}%
\footnote{Some authors use other terms like \emph{synchronizing word} or \emph{directing
word}.%
}. If an automaton has some reset word, we call it a \emph{synchronizing}
\emph{automaton}. A need for finding reset words appears in several
fields of mathematics and engineering. Classical applications (see
\cite{VOL1}) include model-based testing of sequential circuits,
robotic manipulation, and symbolic dynamics, but there are important
connections also with information theory \cite{TRS1} and with formal
models of biomolecular processes \cite{BON1}.

Two particular problems concerning synchronization has gained some
publicity: the Road Coloring Problem and the \v{C}ern\'{y} conjecture.
The first has been solved by Trahtman \cite{TRA6} in 2008 by proving
that the edges of any aperiodic directed multigraph with constant
out-degree (that is, any \emph{admissible} \emph{graph}) can be colored
such that a synchronized automaton arises. Motivation for this problem
comes from symbolic dynamics \cite{ADL1}. On the other hand, the
\v{C}ern\'{y} Conjecture remains open since 1971 \cite{CER1,CER2}.
It claims that any $t$-state synchronized automaton has a reset word
of length at most $\left(t-1\right)^{2}$. 

In the practical applications of synchronization, one may need to
compute a reset word for a given automaton, and moreover, the reset
word should be as short as possible. The need to compute a synchronizing
labeling for an admissible graph, possibly with a request for a short
reset word, may arise as well. It turns out, as we describe below,
that such computational problems are typically NP-hard, even under
various heavy restrictions. 

Parameterized Complexity offers various notions and useful tools that
has became standard in modern analysis of NP-complete problems. The
problems are studied with respect to numerical attributes (\emph{parameters})
of the instances. A \emph{multivariate analysis} considers more than
one such parameter. We close the multivariate analysis for canonical
problems related to synchronization and road coloring. Since the instances
of our problems consist of automata, word lengths and admissible graphs,
the natural parameters are: number of states, alphabet size, and word
length (see the definitions on Page \pageref{defs} and a summary
of the results in Tab. \ref{tab:Complexities-of-SYN} and Tab. \ref{tab:Complexities-of-SRCP}).

In the task to find a synchronized coloring of an admissible graph,
one may also fix a particular reset word to be used in the labeling.
We prove that for the word $abb$ the problem becomes NP-hard though
the corresponding basic variant is decidable in polynomial time (see
Tab. \ref{tab:Complexities-of-SRCPW}).\renewcommand{\arraystretch}{1.6}{\intextsep=19pt
\begin{table}
\begin{centering}
\begin{tabular}{|>{\centering}p{17mm}|>{\centering}p{46mm}>{\raggedleft}p{5mm}|>{\centering}m{43mm}>{\raggedleft}p{5mm}|}
\hline 
\textbf{Parameter} & \multicolumn{2}{>{\centering}m{51mm}|}{\textbf{Parameterized Complexity \qquad{}of~SYN\qquad{}}} & \multicolumn{2}{>{\centering}m{48mm}|}{\textbf{Polynomial Kernel \qquad{}of~SYN\qquad{}}}\tabularnewline
\hline 
$k$ & $\mathrm{W}\!\left[2\right]$-hard & \cite{FER1} & \multicolumn{2}{c|}{\textemdash{}}\tabularnewline
\hline 
$\left|I\right|$ & ~NP-complete for $\left|I\right|=2,3,\dots$ & \cite{EPP1} & \multicolumn{2}{c|}{\textemdash{}}\tabularnewline
\hline 
$k$ and $\left|I\right|$ & FPT, running time $\mathcal{O}^{\star}\!(\left|I\right|^{k})$~~ & \hspace{-9bp}{[}triv\hspace{-1bp}.{]} & ~Not unless $\mathrm{NP}\subseteq\mathrm{coNP}/\mathrm{poly}$~ & \cite{FER1}\tabularnewline
\hline 
$t$ & FPT, running time $\mathcal{O}^{\star}\!(2^{t})$~~ & \hspace{-9bp}{[}triv\hspace{-1bp}.{]} & Not unless PH collapses & $\blacklozenge$\hspace*{2bp}~\tabularnewline
\hline 
\end{tabular}\vspace{2mm}

\par\end{centering}

\begin{centering}
\begin{tabular}{|>{\centering}p{17mm}|>{\centering}p{46mm}>{\raggedleft}p{5mm}|>{\centering}p{43mm}>{\raggedleft}p{5mm}|}
\hline 
\textbf{Parameter} & \multicolumn{2}{>{\centering}m{51mm}|}{\textbf{Parameterized Complexity \qquad{}of~SRCP\qquad{}}} & \multicolumn{2}{>{\centering}m{48mm}|}{\textbf{Polynomial Kernel \qquad{}of~SRCP\qquad{}}}\tabularnewline
\hline 
$k$ & ~NP-complete for $k=4,5,\dots$ & \cite{ROM8} & \multicolumn{2}{c|}{\textemdash{}}\tabularnewline
\hline 
$\left|I\right|$ &  \textbf{~}NP-complete for $\left|I\right|=2,3,\dots$ & $\blacklozenge$\hspace*{2bp}~ & \multicolumn{2}{c|}{\textemdash{}}\tabularnewline
\hline 
$k$ and $\left|I\right|$ & \multicolumn{2}{>{\centering}p{46mm}|}{\emph{See }Tab.\emph{ \ref{tab:Complexities-of-SRCP}}} & \multicolumn{2}{c|}{\textemdash{}}\tabularnewline
\hline 
$t$ & FPT, running time \textbf{$\mathcal{O}^{\star}\!(2^{\left|I\right|})$} & $\blacklozenge$\hspace*{2bp}~ & \qquad{}Yes & $\blacklozenge$\hspace*{2bp}~\tabularnewline
\hline 
\end{tabular}\medskip{}

\par\end{centering}

\caption{\label{tab:Complexities-of-SYN}Results of the complete multivariate
analysis of SYN and SRCP. Diamonds mark the results of the present
paper}
\end{table}

}
\begin{table}[H]
\begin{centering}
\begin{tabular}{|>{\centering}p{20mm}|>{\centering}m{17mm}>{\centering}m{5mm}|>{\centering}m{17mm}>{\centering}m{5mm}|>{\centering}m{17mm}>{\centering}m{5mm}|}
\cline{2-7} 
\multicolumn{1}{>{\centering}p{20mm}|}{} & \multicolumn{2}{c|}{$k=2$} & \multicolumn{2}{c|}{$k=3$} & \multicolumn{2}{c|}{$k=4,5,\dots$}\tabularnewline
\hline 
$\left|I\right|=2$ & \qquad{}P & \cite{ROM8conf} & \qquad{}P & $\blacklozenge$ & \qquad{}NPC & $\blacklozenge$\tabularnewline
\hline 
$\left|I\right|=3$ & \qquad{}P & \cite{ROM8conf} & \qquad{}P & \cite{ROM8} & \qquad{}NPC & \cite{ROM8conf}\tabularnewline
\hline 
$\left|I\right|=4,5,\dots$ & \qquad{}P & \cite{ROM8conf} & \qquad{}P & \cite{ROM8} & \qquad{}NPC & \cite{ROM8conf}\tabularnewline
\hline 
\end{tabular}
\par\end{centering}

\medskip{}

\caption{\label{tab:Complexities-of-SRCP}Complexities of SRCP restricted to
particular values of $k$ and $\left|I\right|$. The cases with $k=1$
or $\left|I\right|=1$ are trivial}
\end{table}

\section{Preliminaries}

\subsection{Automata and Synchronization}

A \emph{deterministic finite automaton }is a triple $A=\left(Q,I,\delta\right)$,
where $Q$ and $I$ are finite sets and $\delta$ is an arbitrary
mapping $Q\times I\rightarrow Q$. Elements of $Q$ are called \emph{states},
$I$ is the \emph{alphabet}. The \emph{transition function} $\delta$
can be naturally extended to $Q\times I^{\star}\rightarrow Q$, still
denoted by $\delta$, slightly abusing the notation. We extend it
also by defining 
\[
\delta\!\left(S,w\right)=\left\{ \delta\!\left(s,w\right)\mid s\in S,w\in I^{\star}\right\} 
\]
 for each $S\subseteq Q$. If an automaton $A=\left(Q,I,\delta\right)$
is fixed, we write
\[
r\overset{x}{\longrightarrow}\, s
\]
instead of $\delta\left(r,x\right)=s$.

For a given automaton $A=\left(Q,I,\delta\right)$, we call $w\in I^{\star}$
a \emph{reset word} if 
\[
\left|\delta\!\left(Q,w\right)\right|=1.
\]
If such a word exists, we call the automaton \emph{synchronizing}.
Note that each word having a reset word as a factor is also a reset
word.

The \emph{\v{C}ern� conjecture}, a longstanding open problem, claims
that each synchronizing automaton has a reset word of length at most
$\left(\left|Q\right|-1\right)^{2}$. There is a series of automata
due to \v{C}ern� whose shortest reset words reach this bound exactly
\cite{CER1}, but all known upper bounds lie in $\OM\left(\left|Q\right|^{3}\right)$.
A tight bound has also been established for various special classes
of automata, see some of recent advances in \cite{GRE1,STE5}. The
best general upper bound of the length of shortest reset words is
currently the following%
\footnote{An improved bound published by Trahtman \cite{TRA1} in 2011 has turned
out to be proved incorrectly.%
}: 
\begin{theorem}
[\cite{PIN2}]\label{thm:z(t)}Any $t$-state synchronizing automaton
has a reset word of length $z\!\left(t\right)$, where
\[
z\!\left(t\right)=\frac{t^{3}-t}{6}.
\]

\end{theorem}
It is convenient to analyze synchronization as a process in discrete
time. Having an automaton $A=\left(Q,I,\delta\right)$ and a word
\[
w=x_{1}\dots x_{\left|w\right|}
\]
fixed, we say that a state $s\in Q$ is \emph{active at time $l\le\left|w\right|$}
if 
\[
s\in\delta\!\left(Q,x_{1}\dots x_{l}\right).
\]
At time $0$, before the synchronization starts, all the states are
active. As we apply the letters, the number of active states may decrease.
We may consider that active states are identified by \emph{activity
markers}, which move along appropriate transitions whenever a letter
is applied. If two activity markers meet each other, they just merge.

When the number of active states decreases to $1$ at a time $l$,
synchronization is complete and the word $x_{1}\dots x_{l}$ is a
reset word.

\subsection{Synchronizing Road Coloring}

A directed multigraph is: 
\begin{enumerate}
\item \emph{aperiodic graph }if the lengths of its cycles do not have any
nontrivial common divisor.
\item \textit{admissible graph }\textit{\emph{if it is aperiodic and all
its out-degrees are equal.}}
\item \textit{road colorable graph }\textit{\emph{if its edges can be labeled
such that a synchronized deterministic finite automaton arises.}}
\end{enumerate}
It is not hard to observe that any road colorable graph is admissible.
In 1977 Adler, Goodwyn and Weiss\cite{ADL1} conjectured that the
backward implication holds as well. Their question became known as
the Road Coloring Problem and a positive answer was given in 2007
by Trahtman \cite{TRA6}:
\begin{theorem}
[Road Coloring Theorem]Any admissible graph is road colorable.
\end{theorem}
In the literature it is common to consider only strongly connected
graphs since many general claims can be easily reduced to the corresponding
claims about strongly connected cases. We do not admit such restriction
explicitly, because in the scope of computational problems it does
not seem very natural. However, all the results hold with the restriction
as well. Especially the NP-completeness proofs have been made slightly
more complicated in order to use strongly connected graphs only.

\subsection{Parameterized Complexity}

In most of the paper, we do not need to work with any formal definition
of a \emph{parameterized problem}. We see it as a classical decision
problem where we consider some special numerical property (\emph{parameter})
of each input. Parameterized complexity is the study of the way in
which the hardness of an NP-complete problem relies on the parameter.
A problem may remain NP-hard even if restricted to instances with
a particular value of the parameter or there may be a distinct polynomial-time
algorithm for each such value (such problems form the class \emph{XP}).
In the second case, if the polynomials are all of the same degree,
we get into the class \emph{FPT}: 

A parameterized problem \emph{is} \emph{fixed-parameter tractable}
(\emph{FPT}) if there is an algorithm that decides it in time 
\[
f\!\left(P\right)\cdot r\!\left(\left|x\right|\right)
\]
where $x$ is the input string, $P\in\mathbb{N}$ its parameter, $r$
is an appropriate polynomial, and $f$ is any computable function.
If there is more than one possible parameter for a problem, one may
consider \emph{combinations} of the parameters. A problem is FPT with
respect to parameters $P,Q$ if it is decidable in time
\[
f\!\left(P,Q\right)\cdot r\!\left(\left|x\right|\right).
\]
This is typically much less restrictive condition than the previous
one, where $f$ depends on $P$ only. 

There is a hierarchy of problems (the \emph{W-hierarchy}) lying in
XP but possibly outside FPT. It consists of the classes $\mathrm{W}\!\left[1\right],\mathrm{W}\!\left[2\right],\dots$:
\begin{equation}
\mathrm{FPT}\subseteq\mathrm{W}\!\left[1\right]\subseteq\mathrm{W}\!\left[2\right]\subseteq\dots\subset\mathrm{XP}.\label{eq:whi}
\end{equation}
Since it has been conjectured that all the inclusions are proper,
it is common to use $\mathrm{W}\!\left[k\right]$-hardness (with respect
to an appropriate type of reduction) as an evidence of lying outside
$\mathrm{FPT}$. However, we do not need to define the W-hierarchy
here since it is used only for the preceding results (see Tab. \ref{tab:Complexities-of-SYN}),
not for the new ones. See the textbook \cite{DF1} for the definitions
and many other great ideas of parameterized complexity.

A \emph{kernel} of a parameterized problem is a polynomial-time procedure
that transforms any input $x$ of the problem to another input $y$
such that the length and the parameter of $y$ are bounded by some
function $f$ of the parameter associated with $x$. Having a kernel
is equivalent to lying in FPT. If the function $f$ is a polynomial,
we get a \emph{polynomial kernel}.

\subsection{Studied Problems}

In this paper we work with two canonical computational problems related
to synchronization (\noun{SYN}) and road coloring (\noun{SRCP}). The
problems are defined as follows:

\renewcommand{\arraystretch}{1.6}

\begin{flushleft}
\label{defs}%
\begin{tabular}{|>{\raggedright}p{25mm}>{\raggedright}p{94mm}|}
\hline 
\multicolumn{2}{|l|}{\noun{~~SYN}}\tabularnewline
\textbf{~~Input:} & Automaton $A=\left(Q,I,\delta\right)$, $k\in\mathbb{N}$\tabularnewline
\textbf{~~Output:} & Is there $w\in I^{\star}$ of length at most $k$ such that $\left|\delta\!\left(Q,w\right)\right|=1$?\tabularnewline[2mm]
~~\textbf{Parameters:} & $k$, $\left|I\right|$, $t=\left|Q\right|$\tabularnewline[2mm]
\hline 
\end{tabular}
\par\end{flushleft}

\begin{flushleft}
\begin{tabular}{|>{\raggedright}p{25mm}>{\raggedright}p{94mm}|}
\hline 
\multicolumn{2}{|l|}{\noun{~~SRCP}}\tabularnewline
\textbf{~~Input:} & Alphabet $I$, admissible graph $G=\left(Q,E\right)$ with out-degrees
$\left|I\right|$, $k\in\mathbb{N}$\tabularnewline
\textbf{~~Output:} & Is there a coloring $\delta$ such that there is $w\in I^{\star}$
of length at most $k$ such that $\left|\delta\!\left(Q,w\right)\right|=1$?\tabularnewline[2mm]
~~\textbf{Parameters:} & $k$, $\left|I\right|$, $t=\left|Q\right|$\tabularnewline[2mm]
\hline 
\end{tabular}
\par\end{flushleft}

We will need the following basic facts related to \noun{SYN}:
\begin{theorem}
[\cite{CER1}]\label{lem:polAlg}There is a polynomial-time algorithm
that decides whether a given automaton is synchronizing.
\end{theorem}

\begin{corollary}
\noun{\label{cor:polAlgSYN}Syn}, if restricted to the instances with
$d\ge z\!\left(t\right)=\frac{n^{3}-n}{6}$, is solvable in polynomial
time.
\end{corollary}

\begin{theorem}
[\cite{EPP1}]\noun{Syn} is NP-complete, even if restricted to automata
with two-letter alphabets. 
\end{theorem}
The results of this paper, as well as the former results of Fernau,
Heggernes, and Villanger \cite{FER1} and of the second author and
Drewienkowski \cite{ROM8,ROM8conf} are summarized by Tables \ref{tab:Complexities-of-SYN},
\ref{tab:Complexities-of-SRCP}, \ref{tab:Complexities-of-SRCPW}.
We have filled all the gaps in the first two tables (cf. corresponding
tables in \cite[Sec. 3]{FER1} and \cite[Sec. 6]{ROM8}), so the multivariate
analysis of \noun{SYN }and \noun{SRCP }is complete in the sense that
NP-complete restrictions are identified and under several standard
assumptions we know which restrictions are FPT and which of them have
polynomial kernels.

\section{Parameterized Complexity of SYN}

The following lemma, which is easy to prove using the construction
of a power automaton, says that \noun{Syn} lies in FPT\emph{ }if parameterized
by number of states:
\begin{lemma}
[\cite{FER1,SAN1}]\label{lem:FPT}There exists an algorithm for deciding
about \noun{SYN} in time $r\!\left(t,\left|I\right|\right)\cdot2^{t}$
for an appropriate polynomial $r$.
\end{lemma}
But does \noun{Syn} have a polynomial kernel? In this section we use
methods developed by Bodlaender et~al. \cite{BOD1} to prove the
following:
\begin{theorem}
If \noun{Syn} has a polynomial kernel, then $\mathrm{PH}=\Sigma_{\mathrm{p}}^{3}$.
\end{theorem}
By $\mathrm{PH}$ we denote the union of the entire polynomial hierarchy,
so $\mathrm{PH}=\Sigma_{p}^{3}$ means that polynomial hierarchy collapses
into the third level, which is widely assumed to be false. The key
proof method relies on \emph{composition algorithms}. In order to
use them immediately, we introduce the formalization of our parameterized
problem as a set of string-integer pairs:
\[
L_{\mathrm{SYN}}=\left\{ \left(x,t\right)\mid x\in\Sigma^{\star}\mbox{ encodes an instance of SYN with }t\in\mathbb{N}\mbox{ states}\right\} ,
\]
where $\Sigma$ is an appropriate finite alphabet.

\subsection{Composition Algorithms}

\label{comp}A \emph{composition algorithm }for a parameterized problem
$L\subseteq\Sigma^{\star}\times\mathbb{N}$ is an algorithm that
\begin{itemize}
\item receives as input a sequence $\left(\left(x_{1},t\right),\dots,\left(x_{m},t\right)\right)$
with $\left(x_{i},t\right)\in\Sigma^{\star}\times\mathbb{N}^{+}$
for each $1\le i\le m$,
\item uses time polynomial in $\sum_{i=1}^{m}\left|x_{i}\right|+t$
\item outputs $\left(y,t'\right)\subseteq\Sigma^{\star}\times\mathbb{N}^{+}$
with 

\begin{enumerate}
\item $\left(y,t'\right)\in L\Leftrightarrow\mbox{there is some }\ensuremath{1\le i\le m}\mbox{ with }\left(x_{i},t\right)\in L$,
\item $t'$ is polynomial in $t$.
\end{enumerate}
\end{itemize}

Let $L\subseteq\Sigma^{\star}\times\mathbb{N}$ be a parameterized
problem. Its \emph{unparameterized version }is 
\[
\widehat{L}=\left\{ x\#a^{t}\mid\left(x,t\right)\in L\right\} .
\]

\begin{theorem}
[\cite{BOD1}]\label{thm:PH}Let $L$ be a parameterized problem having
a composition algorithm. Assume that its unparameterized version $\widehat{L}$
is $\mathrm{NP}$-complete. If $L$ has a polynomial kernel, then
$\mathrm{PH}=\Sigma_{p}^{3}$.
\end{theorem}
The unparameterized version of $L_{\mathrm{SYN}}$ is computationally
as hard as the classical SAT\noun{, }so it is NP-complete. It remains
only to describe a composition algorithm for $L_{\mathrm{SYN}}$,
which is done in the remainder of this section.

\subsection{Preprocessing}

Let the composition algorithm receive an input 
\[
\left(\left(A_{1},d_{1}\right),t\right),\dots,\left(\left(A_{m},d_{m}\right),t\right)
\]
consisting of $t$-state automata $A_{1},\dots,A_{m}$, each of them
equipped with a number $d_{i}$. Assume that the following easy procedures
have been already applied:
\begin{itemize}
\item For each $i=1,\dots,m$ such that $d_{i}\ge z\!\left(t\right)$, use
the polynomial-time synchronizability algorithm from Corollary \pageref{cor:polAlgSYN}
to decide whether $\left(\left(A_{i},d_{i}\right),t\right)\in L_{\mathrm{SYN}}$.
If so, return a trivial true instance immediately. Otherwise just
delete the $i$-th member from the sequence.
\item For each $i=1,\dots,m$, add an additional letter $\kappa$ to the
automaton $A_{i}$ such that $\kappa$ acts as the identical mapping:
$\delta_{i}\!(s,\kappa)=s$.
\item For each $i=1,\dots,m$ rename the states and letters of $A_{i}$
such that 
\begin{eqnarray*}
A_{i} & = & \left(Q_{i},I_{i},\delta_{i}\right)\\
Q_{i} & = & \left\{ 1,\dots,t\right\} \\
I_{i} & = & \left\{ \kappa,a_{i,1},\dots,a_{i,\left|I_{i}\right|-1}\right\} .
\end{eqnarray*}

\end{itemize}
After that, our algorithm chooses one of the following procedures
according to the length $m$ of the input sequence:
\begin{itemize}
\item If $m\ge2^{t}$, use the exponential-time algorithm from Lemma \ref{lem:FPT}:
Denote $D=\sum_{i=1}^{m}\left|\left(A_{i},d_{i}\right)\right|+t$,
where we add lengths of descriptions of the pairs. Note that $D\ge m\ge2^{t}$
and that $D$ is the quantity used to restrict the running time of
composition algorithms. By the lemma, in time
\[
\sum_{i=1}^{m}r\left(t,\left|I_{i}\right|\right)\cdot2^{t}\le m\cdot r\!\left(D,D\right)\cdot2^{t}\le D^{2}\cdot r\!\left(D,D\right)
\]
we are able to analyze all the $m$ automata and decide if some of
them have a reset word of the corresponding length. It remains just
to output some appropriate trivial instance $\left(\left(A',d'\right),t'\right)$.
\item If $m<2^{t}$, we denote $q\!\left(m\right)=\left\lfloor \log\left(m+1\right)\right\rfloor $.
It follows that $q\!\left(m\right)\le t+2$. On the output of the
composition algorithm we put $\left(\left(A',d'\right),t'\right)$,
where $A'$ is the automaton described in the following paragraphs
and 
\[
d'=z\!\left(t\right)+1
\]
is our choice of the maximal length of reset words to be found for
$A'$.
\end{itemize}

\subsection{\label{sub:Construction-of A'}Construction of $A'$ and Its Ideas}

Here we describe the automaton $A'$ that appears in the output of
our composition algorithm. We set
\begin{eqnarray*}
A' & = & \left(Q',I',\delta'\right),\\
Q' & = & \left\{ 1,\dots,t\right\} \cup\left\{ \mathrm{D}\right\} \cup\left(\left\{ 0,\dots,z\!\left(t\right)\right\} \times\left\{ 0,\dots,q\!\left(m\right)\right\} \times\left\{ \mathrm{T},\mathrm{F}\right\} \right),\\
I' & = & \left(\bigcup_{i=1}^{m}I_{i}\right)\cup\left\{ \alpha_{1},\dots,\alpha_{m}\right\} \cup\left\{ \omega_{1},\dots,\omega_{t}\right\} .
\end{eqnarray*}
On the states $\left\{ 1,\dots,t\right\} $ the letters from $\bigcup_{i=1}^{m}I_{i}$
act simply:
\[
\begin{aligned}s\overset{x_{i,j}}{\longrightarrow} & \:\delta_{i}\!\left(s,x_{i,j}\right)\end{aligned}
\]
for each $s\in1,\dots,t$, $i=1,\dots,m$, $j=1,\dots,\left|I_{i}\right|$.
In other words, we let all the letters from all the automata $A_{1},\dots,A_{m}$
act on the states $1,\dots,t$ just as they did in the original automata.
The additional letters act on $\left\{ 1,\dots,t\right\} $ simply
as well:
\[
\begin{aligned}s\overset{\alpha_{i}}{\longrightarrow} & \: s\end{aligned}
\qquad\begin{aligned}s\overset{\omega_{\overline{s}}}{\longrightarrow} & \begin{cases}
\mathrm{D} & \mbox{if }\overline{s}=s\\
s & \mbox{otherwise.}
\end{cases}\end{aligned}
\]
for each $s,\overline{s}\in1,\dots,t$, $i=1,\dots,m$. The state
$\mathrm{D}$ is \emph{absorbing}, which means that
\[
\begin{aligned}\mathrm{D}\overset{y}{\longrightarrow} & \:\mathrm{D}\end{aligned}
\]
for any $y\in I'$. Note that any reset word of $A'$ have to map
all the states of $Q'$ to $\mathrm{D}$. 

The remaining $2\cdot\left(z\!\left(t\right)+1\right)\cdot\left(q\!\left(m\right)+1\right)$
states form what we call a\emph{ guard table}.\emph{ }Its purpose
is to guarantee that:
\begin{enumerate}
\item [(C1)]\label{enu:c123}Any reset word of $A'$ have to be of length
at least $d'=z\!\left(t\right)+1$.
\item [(C2)]Any reset word $w$ of $A'$, having length exactly $z\!\left(t\right)+1$,
is of the form
\begin{equation}
w=\alpha_{i}y_{1}\dots y_{d_{i}}\kappa^{z\!\left(t\right)-1-d_{i}}\omega_{s}\label{eq:form}
\end{equation}
for some $i\in\left\{ 1,\dots,m\right\} $, $y_{1},\dots,y_{d_{i}}\in I_{i}$,
and $s\in\left\{ 1,\dots,t\right\} $, such that $y_{1}\dots y_{d_{i}}$
is a reset word of $A_{i}$. 
\item [(C3)]Any word $w$ 

\begin{itemize}
\item of length $d'=z\!\left(t\right)+1$,
\item of the form (\ref{eq:form}),
\item and satisfying $\delta_{i}\!\left(Q_{i},y_{1}\dots y_{d_{i}}\right)=\left\{ s\right\} $ 
\end{itemize}

is a reset word of $A'$.

\end{enumerate}
If the guard table manages to guarantee these three properties of
$A'$, we are done: Is is easy to check that they imply all the conditions
given in Definition \ref{comp}. So, let us define the action of the
letters from $I'$ on the states from $\left\{ 0,\dots,z\!\left(t\right)\right\} \times\left\{ 0,\dots,q\!\left(m\right)\right\} \times\left\{ \mathrm{T},\mathrm{F}\right\} $.
After that the automaton $A'$ will be complete and we will check
the properties (C1,C2,C3).
\begin{figure}
\begin{centering}
\includegraphics{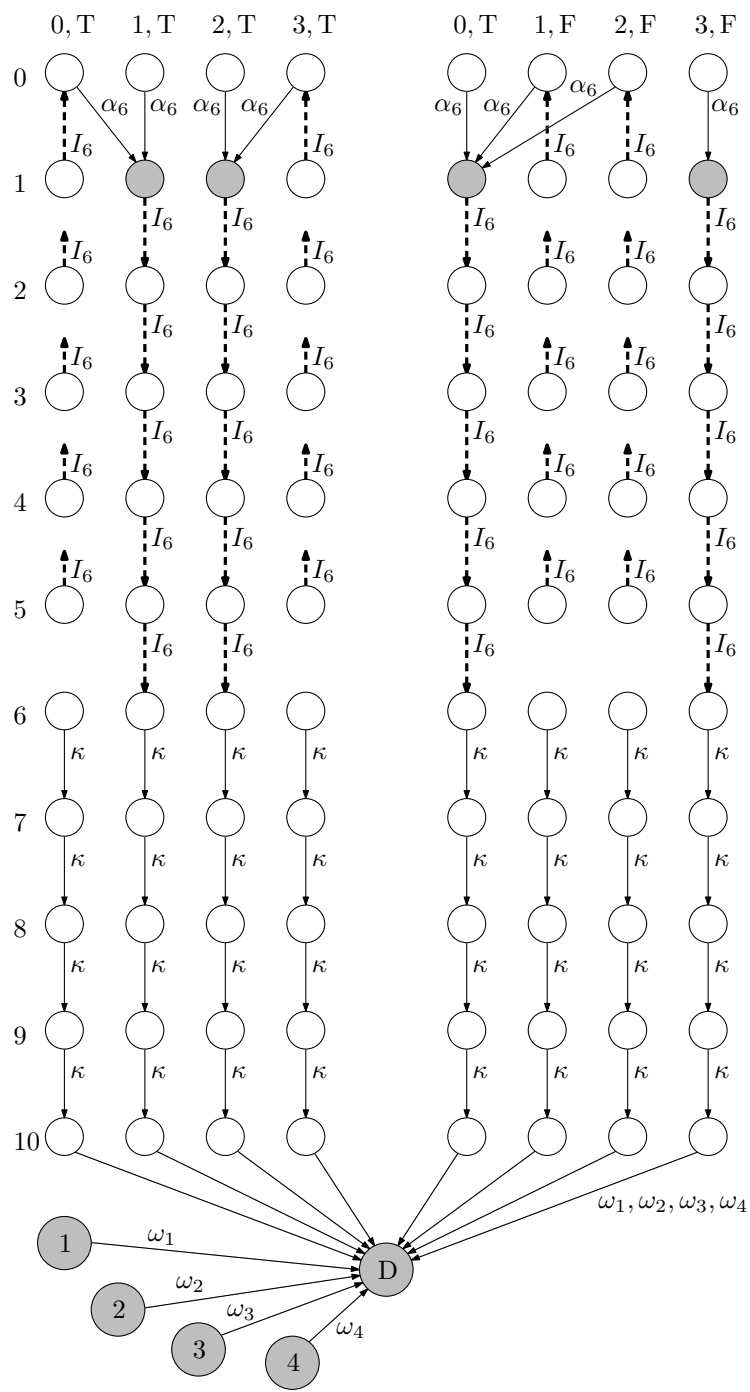}
\par\end{centering}

\caption{\label{fig:Some-transitions-of}Some transitions of the example automaton
described on Page \pageref{Example:}. Grey states remain active after
applying $\alpha_{6}$.}
\end{figure}

The actions of the letters $\alpha_{1},\dots,\alpha_{m}$ should meet
the following two conditions:
\begin{itemize}
\item Any reset word $w$ of length $z\!\left(t\right)+1$ have to start
by some $\alpha_{i}$.
\item In such short reset word, right after the starting $\alpha_{i}$,
there must occur at least $z\!\left(t\right)-1$ consecutive letters
from $I_{i}$. Informally, by applying $\alpha_{i}$ we \emph{choose}
the automaton $A_{i}$.
\end{itemize}
How to do that? The number $m$ may be quite large and each of $\alpha_{1},\dots,\alpha_{m}$
needs to have a unique effect. The key tool is what we call \emph{activity
patterns}. Let us work with the set
\[
R=\left\{ 0,\dots,q\!\left(m\right)\right\} ,
\]
which matches ,,half of a row'' of the guard table. Subsets of $R$
correspond in a canonical way to binary representations of numbers
$0,\dots,2^{q\left(m\right)+1}-1$. We will actually represent only
the numbers $1,\dots,m$. These does not include any of the extreme
values corresponding to the empty set and whole $R$, because we have
$m<2^{q\left(m\right)+1}-1$. So let the mapping
\[
\mathfrak{b}:\left\{ 1,\dots,m\right\} \rightarrow2^{R}
\]
 assign the corresponding subset of $R$ to a number. For instance,
it holds that
\[
\mathfrak{b}\!\left(11\right)=\left\{ 0,1,3\right\} 
\]
 because $11=2^{0}+2^{1}+2^{3}$. For each $i=1,\dots,m$ we define
specific \emph{pattern functions}
\[
\pi_{i}^{\mathrm{T}},\pi_{i}^{\mathrm{F}}:R\rightarrow R
\]
such that
\begin{eqnarray*}
\rng\pi_{i}^{\mathrm{T}} & = & \mathfrak{b}\!\left(i\right),\\
\rng\pi_{i}^{\mathrm{F}} & = & R\backslash\mathfrak{b}\!\left(i\right)
\end{eqnarray*}
for each $i$. It is irrelevant how exactly are $\pi_{i}^{\mathrm{T}}$
and $\pi_{i}^{\mathrm{F}}$ defined. It is sure that they exist, because
the range is never expected to be empty. Now the action of the letters
$\alpha_{1},\dots,\alpha_{m}$ is as follows:
\[
\begin{aligned}\left(h,k,\mathrm{T}\right)\overset{\alpha_{i}}{\longrightarrow} & \:\left(1,\pi_{i}^{\mathrm{T}}\!\left(k\right),\mathrm{T}\right),\\
\left(h,k,\mathrm{F}\right)\overset{\alpha_{i}}{\longrightarrow} & \:\left(1,\pi_{i}^{\mathrm{F}}\!\left(k\right),\mathrm{F}\right).
\end{aligned}
\]
for each $s\in\left\{ 1,\dots,t\right\} $ and each reasonable $h,k$. 

Note that each $\alpha_{i}$ maps the entire guard table, and in particular
the entire row $0$, into the row $1$. In fact, all ,,downward''
transitions within the guard table will lead only one row down. And
the only transitions escaping from the guard table will lead from
the bottom row. Thus any reset word will have length at least $d'=z\!\left(t\right)+1$.
Moreover, during its application, at time $l$ the rows $0,\dots,l-1$
will have to be all inactive. This is a key mechanism that the guard
table uses for enforcing necessary properties of short reset words.

Let us define how the letters $x_{i,j}$ act on the guard table. Choose
any $i\in\left\{ 1,\dots,m\right\} $. The action of $x_{i,j}$ within
the guard table does not depend on $j$, all the letters coming from
a single automaton act identically here:
\begin{itemize}
\item for the rows $h\in\left\{ 1,\dots,d_{i}\right\} $ we set
\end{itemize}
\[
\begin{aligned}\left(h,k,\mathrm{T}\right)\overset{x_{i,j}}{\longrightarrow} & \begin{cases}
\left(h+1,k,\mathrm{T}\right) & \mbox{if }k\in\mathfrak{b}\!\left(i\right)\\
\left(0,k,\mathrm{T}\right) & \mbox{otherwise}
\end{cases}\\
\left(h,k,\mathrm{F}\right)\overset{x_{i,j}}{\longrightarrow} & \begin{cases}
\left(h+1,k,\mathrm{F}\right) & \mbox{if }k\notin\mathfrak{b}\!\left(i\right)\\
\left(0,k,\mathrm{T}\right) & \mbox{otherwise}
\end{cases}
\end{aligned}
\]

\begin{itemize}
\item and for the rows $h\in\left\{ 0\right\} \cup\left\{ d_{i}+1,\dots,z\!\left(t\right)\right\} $
we set
\[
\begin{aligned}\left(h,k,\mathrm{T}\right)\overset{x_{i,j}}{\longrightarrow} & \:\left(0,k,\mathrm{T}\right),\\
\left(h,k,\mathrm{F}\right)\overset{x_{i,j}}{\longrightarrow} & \:\left(0,k,\mathrm{F}\right).
\end{aligned}
\]

\end{itemize}
Recall that sending an activity marker along any transition ending
in the row $0$ is a ,,suicide''. A word that does this cannot be
a short reset word. So, if we restrict ourselves to letters from some
$I_{i}$, the transitions defined above imply that that only at times
$1,\dots,d_{i}$ the forthcoming letter can be some $x_{i,j}$. In
the following $z\!\left(t\right)-d_{i}-1$ steps the only letter from
$I_{i}$ that can be applied is $\kappa$. 

The letter $\kappa$ maps all the states of the guard table simply
one state down, except for the rows $0$ and $z\!\left(t\right)$.
Set
\[
\begin{aligned}\left(h,k,\mathrm{T}\right)\overset{\kappa}{\longrightarrow} & \:\left(h+1,k,\mathrm{T}\right),\\
\left(h,k,\mathrm{F}\right)\overset{\kappa}{\longrightarrow} & \:\left(h+1,k,\mathrm{F}\right).
\end{aligned}
\]
for each $h\in\left\{ 1,\dots,z\!\left(t\right)-1\right\} $, and
\[
\begin{aligned}\left(0,k,\mathrm{T}\right)\overset{\kappa}{\longrightarrow} & \:\left(0,k,\mathrm{T}\right),\\
\left(0,k,\mathrm{F}\right)\overset{\kappa}{\longrightarrow} & \:\left(0,k,\mathrm{F}\right),\\
\left(z\!\left(t\right),k,\mathrm{T}\right)\overset{\kappa}{\longrightarrow} & \:\left(0,k,\mathrm{T}\right),\\
\left(z\!\left(t\right),k,\mathrm{F}\right)\overset{\kappa}{\longrightarrow} & \:\left(0,k,\mathrm{F}\right).
\end{aligned}
\]
It remains to describe actions of the letters $\omega_{1},\dots,\omega_{t}$
on the guard table. Set
\[
\begin{aligned}\left(z\!\left(t\right),k,\mathrm{T}\right)\overset{\omega}{\longrightarrow} & \:\mathrm{D}\\
\left(z\!\left(t\right),k,\mathrm{F}\right)\overset{\omega}{\longrightarrow} & \:\mathrm{D}
\end{aligned}
\]
for each $k$, and
\[
\begin{aligned}\left(h,k,\mathrm{T}\right)\overset{\omega}{\longrightarrow} & \:\left(0,k,\mathrm{T}\right)\\
\left(h,k,\mathrm{F}\right)\overset{\omega}{\longrightarrow} & \:\left(0,k,\mathrm{F}\right)
\end{aligned}
\]
for each $k$ in the remaining rows $h\in\left\{ 0,\dots,z\!\left(t\right)-1\right\} $.
Now the automaton $A'$ is complete.

\subsection{An Example}

\label{Example:}Consider an input consisting of $m=12$ automata
$A_{1},\dots,A_{12}$, each of them having $t=4$ states. Because
$z\!\left(4\right)=10$ and $q\!\left(12\right)=3$, the output automaton
$A'$ has $93$ states in total. In Figure \ref{fig:Some-transitions-of}
all the states are depicted, together with some of the transitions.
We focus on the transitions corresponding to the automaton $A_{6}$,
assuming that $d_{6}=5$. 

The action of $\alpha_{6}$ is determined by the fact that $6=2^{1}+2^{2}$
and thus 
\begin{eqnarray*}
\rng\pi_{6}^{\mathrm{T}} & = & \mathfrak{b}\!\left(6\right)=\left\{ 1,2\right\} ,\\
\rng\pi_{6}^{\mathrm{F}} & = & R\backslash\mathfrak{b}\!\left(6\right)=\left\{ 0,3\right\} .
\end{eqnarray*}
If the first letter of a reset word is $\alpha_{6}$, after its application
only the states
\[
\left(1,1,\mathrm{T}\right),\left(1,2,\mathrm{T}\right),\left(1,0,\mathrm{F}\right),\left(1,3,\mathrm{F}\right)
\]
remain active within the guard table. Now we need to move their activity
markers one row down in each of the following $z\!\left(t\right)-1=9$
steps. The only way to do this is to apply $d_{6}=5$ letters of $I_{6}$
and then $z\!\left(t\right)-1-d_{6}=4$ occurrences of $\kappa$.
Then we are allowed to apply one of the letters $\omega_{1},\dots,\omega_{t}$.
But before that time, there should remain only one active state $s\in\left\{ 1,\dots,t\right\} $,
so that we could use $\omega_{s}$. The letter $\kappa$ does not
affect the activity within $\left\{ 1,\dots,t\right\} $ so we need
to synchronize these states using $d_{6}=5$ letters from $I_{6}$.

So, any short reset word of $A'$ starting with $\alpha_{6}$ has
to contain a short reset word of $A_{6}$.

\subsection{The Guard Table Works}

It remains to use ideas informally outlined in Section \ref{sub:Construction-of A'}
to prove that $A'$ has the properties C1,C2, and C3 from Page \pageref{enu:c123}.
\begin{proof}
[C1]As it has been said, for each letter $x\in I'$ and each state
$\left(h,k,\mathrm{Q}\right)$, where $\mathrm{Q}\in\left\{ \mathrm{T},\mathrm{F}\right\} $
and $h\in\left\{ 0,\dots,z\!\left(t\right)-1\right\} $, it holds
that
\[
\begin{aligned}\left(h,k,\mathrm{Q}\right)\overset{x}{\longrightarrow} & \:\left(h',k',\mathrm{Q}\right),\end{aligned}
\]
where $h'<h$ or $h'=h+1$. So the shortest paths from the row $0$
to the state $\mathrm{D}$ have length at least $z\!\left(t\right)+1$.
\end{proof}

\begin{proof}
[C2]We should prove that any reset word $w$, having length exactly
$z\!\left(t\right)+1$, is of the form
\[
w=\alpha_{i}y_{1}\dots y_{d_{i}}\kappa^{z\!\left(t\right)-1-d_{i}}\omega_{s},
\]
such that, moreover, $y_{1}\dots y_{d_{i}}$ is a reset word of $A_{i}$.
The starting $\alpha_{i}$ is necessary, because $\alpha_{1},\dots,\alpha_{t}$
are the only letters that map states from the row $0$ to other rows.
Denote the remaining $z\!\left(t\right)$ letters of $w$ by $y_{1},\dots,y_{z\left(t\right)}$.

Once an $\alpha_{i}$ is applied, there remain only $\left|R\right|=q\!\left(m\right)+1$
active states in the guard table, all in the row $1$, depending on
$i$. The active states are exactly from 
\[
\left\{ 1\right\} \times\mathfrak{b}\!\left(i\right)\times\left\{ \mathrm{T}\right\} \mbox{ and }\left\{ 1\right\} \times R\backslash\mathfrak{b}\!\left(i\right)\times\left\{ \mathrm{F}\right\} ,
\]
because this is exactly the range of $\alpha_{i}$ within the guard
table. Let us continue by an induction. We claim that for $0\le\tau<d_{i}$
it holds what we have already proved for $\tau=0$:\end{proof}
\begin{enumerate}
\item If $\tau\ge1$, the letter $y_{\tau}$ lies in $I_{i}$. Moreover,
if $\tau>d_{i}$, it holds that $w_{\tau}=\kappa$.
\item After the application of $y_{\tau}$ the active states within the
guard table are exactly from 
\[
\left\{ \tau+1\right\} \times\mathfrak{b}\!\left(i\right)\times\left\{ \mathrm{T}\right\} \mbox{ and }\left\{ \tau+1\right\} \times R\backslash\mathfrak{b}\!\left(i\right)\times\left\{ \mathrm{F}\right\} .
\]

\end{enumerate}
For $i=0$ both the claims hold. Take some $1\le\tau<d_{i}$ and suppose
that the claims hold for $\tau-1$. Let us use the second claim for
$\tau-1$ to prove the first claim for $\tau$. So all the states
from
\[
\left\{ \tau\right\} \times\mathfrak{b}\!\left(i\right)\times\left\{ \mathrm{T}\right\} \mbox{ and }\left\{ \tau\right\} \times R\backslash\mathfrak{b}\!\left(i\right)\times\left\{ \mathrm{F}\right\} 
\]
are active. Which of the letters could appear as $y_{\tau}$? The
letters $\omega_{1},\dots,\omega_{t}$ and $\alpha_{1},\dots,\alpha_{m}$
would map all the active states to the rows $0$ and $1$, which is
a contradiction. Consider any letter $x_{k,j}$ for $k\neq i$. It
holds that $\mathfrak{b}\!\left(i\right)\neq\mathfrak{b}\!\left(k\right)$,
so there is some $c\in R$ lying in their symmetrical difference.
For such $c$ it holds that 
\[
\left(\tau,c,\mathrm{T}\right)\overset{x_{k,j}}{\longrightarrow}\,\left(0,c,\mathrm{T}\right)\mbox{ if }c\in\mathfrak{b}\!\left(i\right)\backslash\mathfrak{b}\!\left(k\right)
\]
or
\[
\left(\tau,c,\mathrm{F}\right)\overset{x_{k,j}}{\longrightarrow}\,\left(0,c,\mathrm{F}\right)\mbox{ if }c\in\mathfrak{b}\!\left(k\right)\backslash\mathfrak{b}\!\left(i\right)
\]
which necessarily activates some state in the row $0$, which is a
contradiction again. So, $y_{\tau}\in I_{i}$. Moreover, if $\tau>d_{i}$,
the letters from $I_{i}\backslash\left\{ \kappa\right\} $ map the
entire row $\tau$ into the row $0$, so the only possibility is $y_{\tau}=\kappa$. 

The letter $y_{\tau}$ maps all the active states right down to the
row $\tau+1$, so the second claim for $\tau$ holds as well.

\begin{proof}
[C3]It is easy to verify that no ,,suicidal'' transitions within
the guard table are used, so during the application of 
\[
y_{1}\dots y_{d_{i}}\kappa^{z\!\left(t\right)-1-d_{i}}
\]
 the activity markers just flow down from the row $1$ to the row
$z\!\left(t\right)$. Since $y_{1}\dots y_{d_{i}}$ is a reset word
of $A_{i}$, there also remains only one particular state $s$ within
$\left\{ 1,\dots,t\right\} $. Finally the letter $\omega_{s}$ is
applied which maps $s$ and the entire row $z\!\left(t\right)$ directly
to $\mathrm{D}$.
\end{proof}

\section{Parameterized Complexity of SRCP}

\subsection{Parameterization by the Number of States}

We point out that SRCP parameterized by the number of states has a
polynomial kernel, so it necessarily lies in FPT.
\begin{theorem}
There is a polynomial kernel for \noun{SRCP} parameterized by $t=\left|Q\right|$.\end{theorem}
\begin{proof}
The algorithm takes an instance of SRCP, i.e. an alphabet $I$, an
admissible graph $G=\left(Q,E\right)$ with out-degrees $\left|I\right|$,
and a number $k\in\mathbb{N}$. It produces another instance of size
depending only on $t=\left|Q\right|$. If $k\ge z\!\left(t\right)$,
we just solve the problem using Corollary \pageref{cor:polAlgSYN}
and output some trivial instance. Otherwise the output instance is
denoted by $I',G'=\left(Q',E'\right),k'$ where
\begin{eqnarray*}
Q' & = & Q\\
k' & = & k\\
\left|I'\right| & = & \min\left\{ \left|I\right|,t\cdot\left(z\!\left(t\right)-1\right)\right\} 
\end{eqnarray*}
and the algorithm just deletes appropriate edges in order to reduce
the out-degree to $\left|I'\right|$. Let us use a procedure that:\end{proof}
\begin{itemize}
\item takes an admissible graph with out-degree $d>t\cdot\left(z\!\left(t\right)-1\right)$
\item for each of its vertices:

\begin{itemize}
\item finds an outgoing multiedge with the largest multiplicity (which is
at least $z\!\left(t\right)$)
\item deletes one edge from the multiedge
\end{itemize}
\end{itemize}
Clearly the resulting graph has out-degree $d-1$. We create the graph
$G'$ by repeating this procedure (starting with $G$) until the out-degree
is at most $t\cdot\left(z\!\left(t\right)-1\right)$.

Now we claim that
\begin{eqnarray*}
 & \left(I,G,k\right)\in\mathrm{SRCP}\\
 & \Updownarrow\\
 & \left(I',G',k'\right)\in\mathrm{SRCP}.
\end{eqnarray*}
The upward implication is trivial since any coloring of $G'$ can
be extended to $G$ and the appropriate reset word can be still used.
On the other hand, let us have a coloring $\delta$ of $G$ such that
$\left|\delta\!\left(Q,w\right)\right|=1$ for a word $w$ of length
at most $k<z\!\left(t\right)$, so it uses at most $z\!\left(t\right)-1$
letters from $I$. If we delete from $G$ all the edges labeled by
non-used letters, we get a subgraph of $G'$ because during the reduction
of edges we have reduced only multiedges having more than $z\!\left(t\right)-1$
edges. So we are able to color $G'$ according to the used letters
of $G$ and synchronize it by the word $w$.
\begin{corollary}
\noun{SRCP} parameterized by $t=\left|Q\right|$ lies in \noun{FPT}.
\end{corollary}

\subsection{Restriction to $\left|I\right|=2$ and $k=3$}

Here we prove that \noun{SRCP} restricted to $\left|I\right|=2$ and
$k=3$ is decidable in polynomial time. If $G=\left(Q,E\right)$ is
a graph, by $V_{i}\!\left(q\right)$ we denote the set of vertices
from which there is a path of length $i$ leading to $q$ and there
is no shorter one. For any $w\in I^{*}$, $\mathbb{G}_{w}$ denotes
the set of graphs with outdegree $2$ that admit a coloring $\delta$
such that $\delta\!\left(Q,w\right)=\left\{ q\right\} $ for some
$q\in Q$.
\begin{lemma}
\label{lem:Let-.-Then}Let $G=\left(Q,E\right)\in\mathbb{G}_{abb}\backslash\mathbb{G}_{aaa}$.
Then some of the following conditions hold:
\begin{enumerate}
\item There is a vertex $q\in Q$ such that each vertex has an outgoing
edge leading into $V_{2}\!\left(q\right)$.
\item $G\in\mathbb{G}_{aba}$
\end{enumerate}
\end{lemma}
\begin{proof}
Let $G=\left(Q,E\right)\in\mathbb{G}_{abb}\backslash\mathbb{G}_{aaa}$.
So $G$ admits a coloring $\delta$ such that 
\[
\delta\!\left(Q,abb\right)=\left\{ q\right\} 
\]
for a state $q\in Q$.\end{proof}
\begin{itemize}
\item If the coloring $\delta$ satisfies $q\notin\delta\!\left(Q,a\right)$,
notice that each edge labeled by $a$ have to lead into $V_{2}\!\left(q\right)$.
Indeed:

\begin{itemize}
\item It cannot lead to $q$ due to $q\notin\delta\!\left(Q,a\right)$.
\item It cannot lead into $V_{1}\!\left(q\right)$ because in such case,
using $q\notin\delta\!\left(Q,a\right)$, it would hold that $q\in\delta\left(Q,ab\right)$,
so it would be necessary to have $\delta\left(q,b\right)=q$, but
from $G\notin\mathbb{G}_{aaa}$ it follows that there is no loop on
$q$.
\item It cannot lead to $V_{3}\!\left(q\right)$, because there is no path
of length $2$ from $V_{3}\!\left(q\right)$ to $q$.
\end{itemize}

So the condition (1) holds.

\item Otherwise the coloring $\delta$ satisfies $q\in\delta\!\left(Q,a\right)$.
Denote 
\[
W=\left\{ s\in Q\mid\mbox{in }\delta\mbox{ there is an edge }s\overset{b}{\longrightarrow}q\right\} .
\]
Now define another coloring $\delta'$ by switching the colors of
the two edges leaving each state of $W$. We claim that
\[
\delta'\!\left(Q,aba\right)=\left\{ q\right\} 
\]
and so the condition (2) holds. Indeed:

\begin{itemize}
\item Take $s\in V_{3}\!\left(q\right)$. In $\delta$ there is a path
\begin{equation}
s\overset{a}{\longrightarrow}t\overset{b}{\longrightarrow}u\overset{b}{\longrightarrow}q.\label{eq: path1}
\end{equation}
Because $s\in V_{3}\!\left(q\right)$, it holds that $t\in V_{2}\!\left(q\right)$
and $u\in V_{1}\!\left(q\right)$. It follows that $t\notin W,u\in W$
and thus in $\delta'$ there is a path
\begin{equation}
s\overset{a}{\longrightarrow}t\overset{b}{\longrightarrow}u\overset{a}{\longrightarrow}q.\label{eq:path2}
\end{equation}

\item Take $s\in V_{2}\!\left(q\right)$. In $\delta$ there is a path (\ref{eq: path1}). 

\begin{itemize}
\item If $t\in V_{2}\!\left(q\right)$, we get again that $t\notin W,u\in W$
and thus in $\delta'$ there is a path (\ref{eq:path2}). 
\item Otherwise we have $t\in V_{1}\!\left(q\right)$. Because $G\notin\mathbb{G}_{aaa}$,
there is no loop on $q$, thus $u\neq q$ and thus $t\notin W$. But
$u\in W$, so we get a path (\ref{eq:path2}) again.
\end{itemize}
\item Take $s\in V_{1}\!\left(q\right)$. In $\delta'$ there is always
an edge $s\overset{a}{\longrightarrow}q$, so we need just $\delta'\!\left(q,ba\right)=q$.
Because we assume that $q\in\delta\!\left(Q,a\right)$, in $\delta$
there have to be a cycle $q\overset{b}{\longrightarrow}r\overset{b}{\longrightarrow}q$
for some $r\in V_{1}\!\left(q\right)$. In $\delta'$ we have $q\overset{b}{\longrightarrow}r\overset{a}{\longrightarrow}q$.
\item For $s=q$ we apply the same reasoning as for $s\in V_{2}\!\left(q\right)$. 
\end{itemize}
\end{itemize}
\begin{theorem}
\label{thm:For-each-}For each $G$ with outdegree $2$ it holds that
\[
G\in\mathbb{G}_{abb}\backslash\left(\mathbb{G}_{aba}\cup\mathbb{G}_{aaa}\right)
\]
if and only if
\begin{itemize}
\item \textup{It holds that $G\notin\mathbb{G}_{aba}\cup\mathbb{G}_{aaa}$.}
\item \textup{There is a vertex $q\in Q$ such that each vertex has an outgoing
edge leading into $\ensuremath{V_{2}\!\left(q\right)}$.}
\end{itemize}
\end{theorem}
\begin{proof}
The downward implication follows easily from Lemma \ref{lem:Let-.-Then}.
For the upward one we need only to deduce that $G\in\mathbb{G}_{abb}$.
We construct the following coloring $\delta$:\end{proof}
\begin{itemize}
\item The edges leading into $V_{2}\!\left(q\right)$ are labeled by $a$.
If two such edges start in a common vertex, they are labeled arbitrarily.
\item The other edges are labeled by $b$.
\end{itemize}
This works because from any state $s\in V_{2}\!\left(q\right)$ there
is an edge leading to some $t\in V_{1}\!\left(q\right)$, and from
$t$ there is an edge leading to $q$. We have labeled both these
edges by $b$. It follows that wherever we start, the path labeled
by $abb$ leads to $q$.
\begin{theorem}
SRCP with $l=2$ and $C=3$ lies in P.\end{theorem}
\begin{proof}
Let the algorithm test the membership of a given graph $G$ for the
following sets:\end{proof}
\begin{enumerate}
\item $\mathbb{G}_{aaa}$,
\item $\mathbb{G}_{aab}\backslash\mathbb{G}_{aaa}$,
\item $\mathbb{G}_{aba}\backslash\mathbb{G}_{aaa}$,
\item $\mathbb{G}_{abb}\backslash\left(\mathbb{G}_{aba}\cup\mathbb{G}_{aaa}\right)$.
\end{enumerate}
For the sets 1,2,3 the membership is polynomially testable due to
results from \cite{ROM8}. For the set $4$ we have proved it by Theorem
\ref{thm:For-each-}. It is easy to see that a graph $G$ should be
accepted if and only if it lies in some of the sets.

\subsection{Restriction to $\left|I\right|=2$ and $k=4$}
\begin{theorem}
SRCP remains NP-complete if restricted to $\left|I\right|=2$ and
$k=4$.\end{theorem}
\begin{proof}
Let us perform a reduction from 3-SAT. Consider a propositional formula
of the form
\[
\Phi=\bigwedge_{j=1}^{m}\mathcal{C}_{j}
\]
where
\[
\mathcal{C}_{j}=l_{i,1}\vee l_{i,2}\vee l_{i,3}
\]
and 
\[
l_{j,k}\in\left\{ x_{1},\dots,x_{n},\overline{x_{1}},\dots,\overline{x_{n}}\right\} 
\]
for each $j=1,\dots,m$ and $k=1,2,3$.

We construct a directed multigraph $G_{\Phi}=\left(Q,E\right)$ with
\[
\left|Q\right|=5m+3n+8
\]
 states, each of them having exactly two outgoing edges. We describe
the set $Q$ as a disjoint union of the sets
\[
Q=\mathbf{C}_{1}\cup\dots\cup\mathbf{C}_{m}\cup\mathbf{V}_{1}\cup\dots\cup\mathbf{V}_{n}\cup\mathbf{D},
\]
where
\begin{eqnarray*}
\mathbf{C}_{j} & = & \left\{ \mathrm{C}_{j,0},\mathrm{C}_{j,1},\mathrm{C}_{j,2},\mathrm{C}_{j,3},\mathrm{C}_{j,4}\right\} ,\\
\mathbf{V}_{i} & = & \left\{ x_{i},\overline{x_{i}},\mathrm{W}_{i}\right\} ,\\
\mathbf{D} & = & \left\{ \mathrm{D}_{0},\dots,\mathrm{D}_{7}\right\} ,
\end{eqnarray*}
for each $j=1,\dots,m$ and $i=1,\dots,n$. The parts $\mathbf{C}_{j}$
correspond to clauses, the parts $\mathbf{V}_{i}$ correspond to variables.
In each $\mathbf{V}_{i}$ there are two special states labeled by
literals $x_{i}$ and $\overline{x_{i}}$. All the edges of $G_{\Phi}$
are defined by Figures \ref{fig:C}, \ref{fig:V}, \ref{fig:D}. 
\begin{figure}
\begin{minipage}[t]{0.45\columnwidth}%
\begin{center}
\includegraphics[scale=0.8]{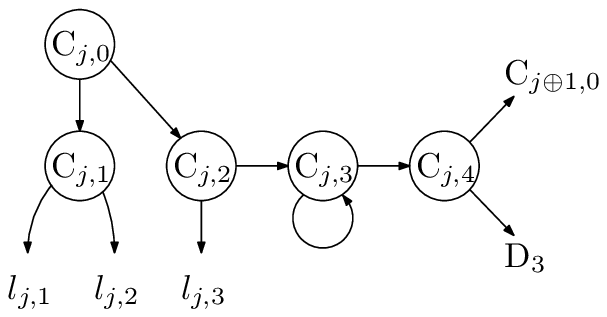}\caption{\label{fig:C}A part $\mathbf{C}_{j}$. Note the three edges that
depend on $\Phi$: they end in vertices labeled by literals from $\mathcal{C}_{j}$.}

\par\end{center}%
\end{minipage}\hfill{}%
\begin{minipage}[t]{0.45\columnwidth}%
\begin{center}
\includegraphics[scale=0.8]{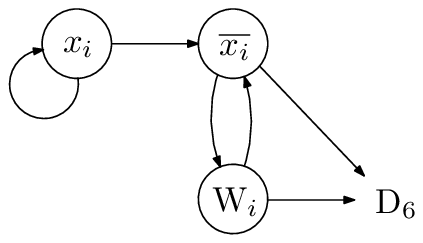}\caption{\label{fig:V}A part $\mathbf{V}_{i}$}

\par\end{center}%
\end{minipage}
\end{figure}
\begin{figure}
\begin{centering}
\includegraphics{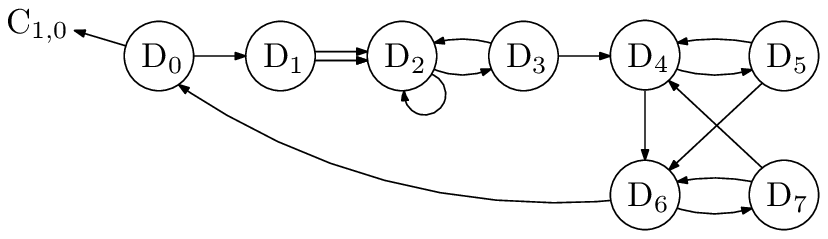}
\par\end{centering}

\caption{\label{fig:D}The part $\mathbf{D}$}
\end{figure}
\begin{figure}
\begin{centering}
\includegraphics{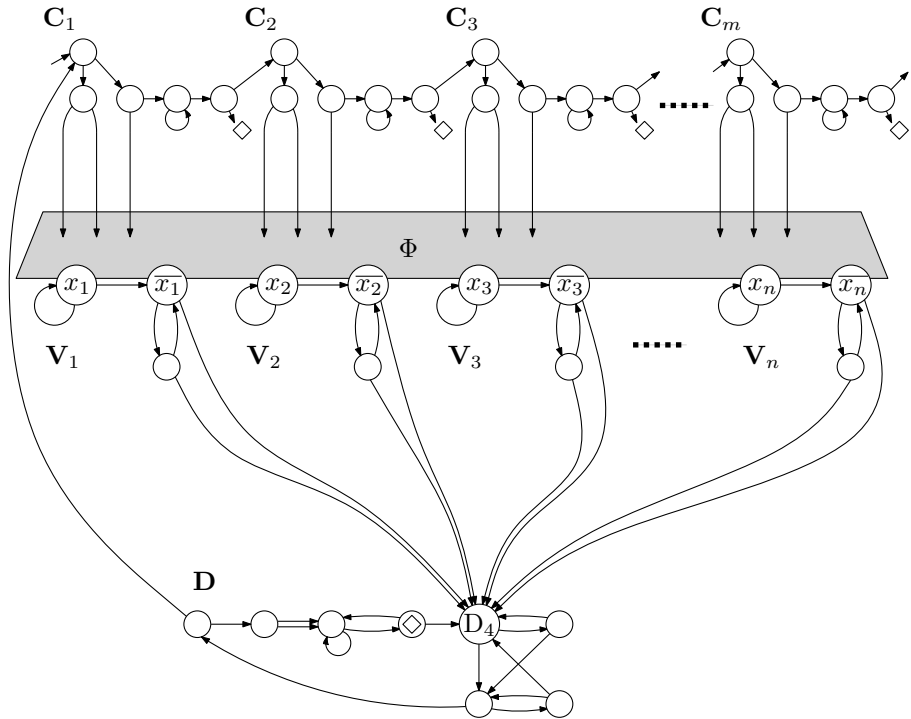}
\par\end{centering}

\caption{\label{fig:All}The entire $G_{\Phi}$}
\end{figure}
Figure \ref{fig:All} gives an overall picture of $G_{\Phi}$. Let
us prove that
\begin{eqnarray*}
 & \Phi\mbox{ is satisfiable}\\
 & \Updownarrow\\
 & G_{\Phi}\mbox{ can be synchronized by some word of length }4\mbox{ for some labeling}
\end{eqnarray*}

\end{proof}

\subsubsection*{The Upward Implication}
\begin{proof}
Suppose that there is a labeling $\delta$ by letters $a$ (solid)
and $b$ (dotted) such that there is a word 
\[
w=y_{1}\dots y_{4}\in\left\{ a,b\right\} ^{4}
\]
 with
\[
\left|Q.w\right|=1.
\]
Let $a$ be the first letter of $w$. By \emph{$k$-path} (resp. $k$\emph{-reachable})
we understand path of length exactly $k$ (resp. reachable by a path
of length exactly $k$).
\begin{figure}[H]
\begin{centering}
\includegraphics{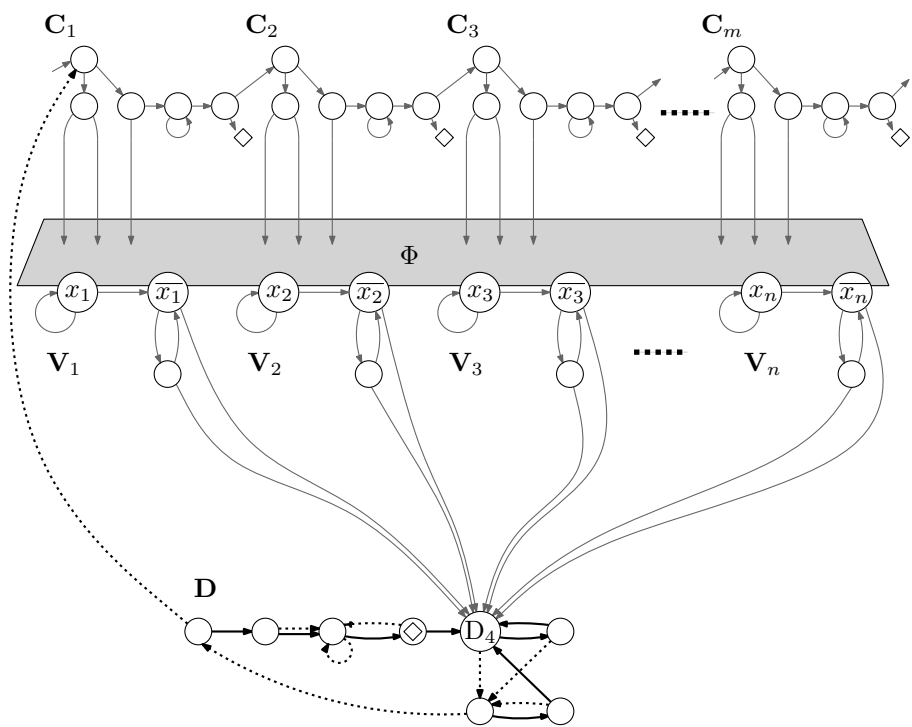}
\par\end{centering}

\caption{\label{fig:All-1}The entire $G_{\Phi}$ with the edges outgoing from
$\mathrm{D}$ colored. Bold arrows: $a$, dotted arrows: $b$.}
\end{figure}
\end{proof}
\begin{lemma}
The synchronization takes place in $\mathrm{D}_{4}$.\end{lemma}
\begin{proof}
From $\mathrm{D}_{1}$ only states from $\mathbf{D}$ are $4$-reachable.
From $\mathrm{D_{0}}$ the only states within $\mathbf{D}$ that are
$4$-reachable are $\mathrm{D}_{2},\mathrm{D}_{3},\mathrm{D}_{4}$.
From $\mathrm{C}_{0,1}$ only $\mathrm{D_{4}}$ is $4$-reachable.\end{proof}
\begin{lemma}
All edges outgoing from states of $\mathbf{D}$ are labeled as in
Figure \ref{fig:All-1}.\end{lemma}
\begin{proof}
Since $\mathrm{D}_{4}$ is not $3$-reachable from $\mathrm{D}_{0}$
nor $\mathrm{D}_{6}$, all the edges incoming to $\mathrm{D}_{0}$
and $\mathrm{D}_{6}$ are labeled by $b$. The remaining labeling
follows easily.\end{proof}
\begin{corollary}
It holds that 
\[
w=aba^{2}.
\]
\end{corollary}
\begin{lemma}
For each $j=1,\dots,m$ we have
\[
\mathrm{C}_{j,0}.ab\in\left\{ x_{1},\dots,x_{n},\overline{x_{1}},\dots,\overline{x_{n}}\right\} .
\]
\end{lemma}
\begin{proof}
Any of the other states $2$-reachable from $\mathrm{C}_{j,0}$ does
not offer a $2$-path leading to $\mathrm{D}_{4}$.
\begin{figure}
\begin{centering}
\includegraphics{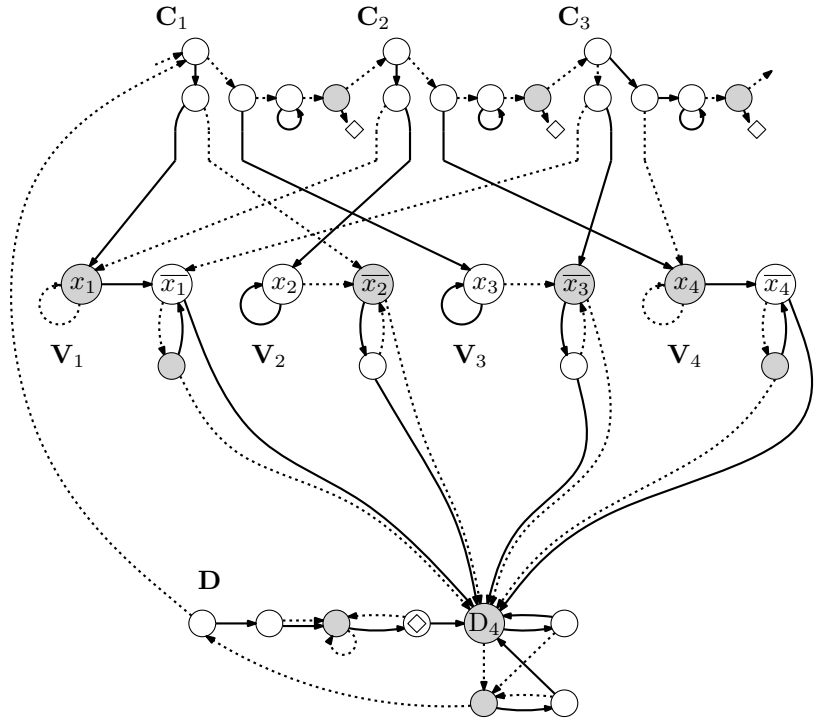}
\par\end{centering}

\caption{\label{fig:All-2}An example of $G_{\Phi}$ for $\Phi=\left(\ensuremath{x_{1}\vee\overline{x_{2}}\vee x_{3}}\right)\wedge\left(\ensuremath{x_{1}\vee x_{2}\vee x_{4}}\right)\wedge\left(\ensuremath{\overline{x_{1}}\vee\overline{x_{3}}\vee x_{4}}\right)$.\protect \\
The filling marks states that are active after applying $y_{1}y_{2}=ab$.}
\end{figure}
\end{proof}
\begin{lemma}
There are no $j,k\in1,\dots,m$ and $i\in1,\dots,n$ such that 
\[
\mathrm{C}_{j,0}.ab=x_{i}
\]
and
\[
C_{k,0}.ab=\overline{x_{i}}.
\]
\end{lemma}
\begin{proof}
If both $x_{i}$ and $\overline{x_{i}}$ are active after applying
$y_{1}y_{2}=ab$, there have to be $2$-paths labeled by $a^{2}$
from both the states $x_{i},\overline{x_{i}}$ to $\mathrm{D}_{4}$.
It is easy to see that it is not possible to find such labeling.\end{proof}
\begin{corollary}
There is a partial assignment making all the literals
\[
\mathrm{C}_{1,0}.ab,\mathrm{C}_{2,0}.ab,\dots,\mathrm{C}_{m,0}.ab
\]
satisfied, because none of them is the negation of another. Each clause
contains some of these literals.
\end{corollary}
We are done, the existence of a satisfying assignment is guaranteed.\qed

\subsubsection*{The Downward Implication}

For a given satisfying assignment we make a coloring based on the
above-mentioned ideas and the example given by Fig. \ref{fig:All-2}. 
\begin{itemize}
\item For each $j$, the coloring of edges outgoing from $\mathrm{C}_{j,0},\mathrm{C}_{j,1},\mathrm{C}_{j,2}$
depends on which of the three literals of the clause $\mathcal{C}_{j}$
are satisfied by the assignment (the example assigns $x_{1}=\mathbf{1},x_{2}=\mathbf{0},x_{3}=\mathbf{0},x_{4}=\mathbf{1}$).
The $2$-path from $\mathrm{C}_{j,0}$ labeled by $ab$ should lead
to a state labeled by a satisfied literal. The edges outgoing from
$\mathrm{C}_{j,3}$ and $\mathrm{C}_{j,4}$ are colored in a constant
way. 
\item For each $i$, all the edges outgoing from the states of the $\mathbf{V}_{i}$
part are colored in one of two ways depending on the truth value assigned
to $x_{i}$. 
\item The edges outgoing from the states of $\mathbf{D}$ admit the only
possible coloring.
\end{itemize}
Note that in our example the edges outgoing from the states of $\mathbf{V}_{3}$
could be colored in the opposite way as well. None of the literals
$x_{3},\overline{x_{3}}$ is chosen by the coloring to satisfy a clause.

\subsubsection*{Strong Connectivity}

If there is a non-negated occurrence of each $x_{i}$ in $\Phi$,
the graph $G_{\Phi}$ is strongly connected. This assumption can be
easily guaranteed by adding tautological clauses like $x_{i}\vee\overline{x_{i}}\vee\overline{x_{i}}$.

\section{Further Research: SRCPW}

On the input of SRCP there is a prescribed length of a reset word
that should be used in the road coloring. But what if an exact reset
word (or a set of possible reset words) were prescribed? We call such
problem SRCPW:

\begin{flushleft}
\begin{tabular}{|>{\raggedright}p{25mm}>{\raggedright}p{94mm}|}
\hline 
\multicolumn{2}{|l|}{\noun{~~SRCPW}}\tabularnewline
\textbf{~~Input:} & Alphabet $I$, admissible graph $G=\left(Q,E\right)$ with out-degrees
$\left|I\right|$, $W\subseteq I^{\star}$\tabularnewline
\textbf{~~Output:} & Is there a coloring $\delta$ such that $\left|\delta\!\left(Q,w\right)\right|=1$
for some $w\in W$?\tabularnewline[2mm]
~~\textbf{Parameters:} & $\left|I\right|$, $t=\left|Q\right|$, $W\subseteq I^{\star}$ (non-numerical)\tabularnewline[2mm]
\hline 
\end{tabular}
\par\end{flushleft}

We have found out that even when restricted to $\left|I\right|=2$,
the fixed value $W=\left\{ abb\right\} $ makes the problem NP-complete.
This may seem quite surprising because we have shown above that SRCP
restricted to $\left|I\right|=2$ and $k=3$ is polynomially decidable.
\begin{theorem}
[forthcoming paper]\noun{SRCPW }restricted to $\left|I\right|=2$
and $W=\left\{ abb\right\} $ is NP-complete.
\end{theorem}
Together with related results from \cite{ROM8} we get the situation
depicted by Table \ref{tab:Complexities-of-SRCPW}. Clearly there
is a wide range of open problems about SRCPW and its restrictions
to particular values of $W$.

\begin{table}[H]
\begin{centering}
\begin{tabular}{|>{\centering}p{20mm}|>{\centering}m{17mm}>{\centering}m{5mm}|>{\centering}m{17mm}>{\centering}m{5mm}|}
\cline{2-5} 
\multicolumn{1}{>{\centering}p{20mm}|}{} & \multicolumn{2}{c|}{$\left|I\right|=2$} & \multicolumn{2}{c|}{$\left|I\right|=3$}\tabularnewline
\hline 
$W=\left\{ aaa\right\} $ & \qquad{}P &  & \qquad{}P & \tabularnewline
\hline 
$W=\left\{ aab\right\} $ & \qquad{}P &  & \qquad{}P & \tabularnewline
\hline 
$W=\left\{ aba\right\} $ & \qquad{}P &  & \qquad{}P & \tabularnewline
\hline 
$W=\left\{ abb\right\} $ & \qquad{}NPC & $\blacklozenge$ & \qquad{}Open & \tabularnewline
\hline 
$W=\left\{ abc\right\} $ & \qquad{}\textemdash{} &  & \qquad{}Open & \tabularnewline
\hline 
\end{tabular}
\par\end{centering}

\medskip{}

\caption{\label{tab:Complexities-of-SRCPW}Complexities of SRCPW restricted
to particular values of $W$ and $\left|I\right|$. The positive results
come from \cite{ROM8}.}
\end{table}

\bibliographystyle{C:/Users/Vojta/Desktop/SYNCHRO2/splncs03}
\bibliography{C:/Users/Vojta/Desktop/SYNCHRO2/bib/ruco}

\end{document}